\documentclass{llncs}

\usepackage{amsmath}
\usepackage{amsfonts}
\usepackage{amssymb}
\usepackage{color}

\usepackage{tikz}
\usetikzlibrary{shapes.geometric}
\usetikzlibrary{decorations.markings}
\tikzset{->-/.style={decoration={markings, mark=at position .5 with {\arrow{>}}}, postaction={decorate}}}

\usepackage{hyperref}

\renewcommand\epsilon\varepsilon

\newcommand\id{\mathrm{id}} % identity
\renewcommand\P{\mathbf P} % play
\newcommand\C{\mathbf C} % coplay
\newcommand\E{\mathbf E} % equilibrium function
\newcommand\argmax{\arg\max}
\newcommand\fix{\mathrm{fix}}
\newcommand\R{\mathbb R}

\newcommand\bangrot{\text{\textexclamdown}} % upside down !

\newcommand\G{\mathcal G}
\renewcommand\H{\mathcal H}

\renewcommand\Game{\mathbf{Game}}

\title{The algebra of predicting agents}
\author{Joe Bolt \and Jules Hedges \and Viktor Winschel}
\institute{University of Oxford}

\begin{document}

\maketitle

%\textcolor{red}{comments in red. }\textcolor{blue}{suggested revisions in
%blue.}

\begin{abstract}
	The category of open games, which provides a
	strongly compositional foundation of economic game theory, is
	intermediate between symmetric monoidal and compact closed. More precisely it
	has counits with no corresponding units, and a partially defined
	duality. There exist open games with the same types as unit maps,
	given by agents with the strategic goal of predicting a future
	value. Such agents appear in earlier work on selection functions.
	We explore the algebraic properties of these agents via the symmetric monoidal bicategory whose 2-cells are morphisms between open games, and show how
	the resulting structure approximates a compact closed category with
	a family of lax commutative bialgebras.
\end{abstract}

\section{Introduction}

Open games \cite{hedges15b} provide a strongly compositional foundation to economic game theory.
In this paper we continue the investigation of the categorical structure of open games.

An open game is, in general, a fragment of a game that can be embedded in a context of an appropriate type.
Open games are the morphisms of a symmetric monoidal
category and, hence, admit categorical composition and monoidal product
operators. These represent sequential and simultaneous play
respectively.
Games are built compositionally, beginning with simple `atomic' open games
such as individual decisions and payoff functions, using the composition
operators.
This is closely related to the categorical open systems research programme \cite{fong_algebra_open_interconnected_systems}.

A related line of work has investigated an approach to game theory based on
\emph{selection functions} (see for example \cite{escardo10a,escardo11}).
These come in two flavours: \emph{single-valued} selection functions have a type of the form $(X \to R) \to X$, and \emph{multi-valued} selection functions $(X \to R) \to \mathcal P (X)$.
These allow the replacement of the $\argmax$ operator, which is itself a selection function of type $(X \to \mathbb R) \to \mathcal P (X)$, with other operators
with similar types.
In particular, a worked example in \cite{hedges_etal_selection_equilibria_higher_order_games} considers games in which $\argmax$ is replaced by a fixpoint operator of type $(X \to X) \to \mathcal P (X)$, defined by $\fix (k) = \{ x \mid x = k (x) \}$.
This provides an elegant toy model of \emph{coordination} as a strategic aim.

Any selection function can be viewed as an open game, representing a single decision made by an agent whose `rationality' is defined by that selection function.
In particular, the $\argmax$ selection function defines a `classically rational' agent, which is one of the atomic open games.
In this paper we consider the agents that arise from the fixpoint selection function.
The strategic goal of these agents is to correctly predict a future value.

The category of open games admits a `partial duality' $-^*$ that is defined
on all objects but on few morphisms, and there are `counit' open games
$\epsilon_X : X \otimes X^* \to I$ that are compatible with this duality
whenever it is defined.
In this sense, the categorical structure of open games is similar to a
`fragment' of compact closed structure \cite{selinger11}.
% \textcolor{red}{I'm not sure about the
%phrasing ``fragment' of compact closed.' Compact closure is a property -
%can something be a fragment of a property? How about `almost compact
%closed?', 'approaching compact closed?'}
The fixpoint agent has the type $\eta_X : I \to X \otimes X^*$, and due to its interpretation as a past prediction of a future value, is a natural candidate to be the dual `unit'.

Unfortunately we do not obtain a compact closed category.
We investigate what structure we do obtain, using the symmetric monoidal bicategory whose 2-cells are globular morphisms of open games, obtained as a special case of more general morphisms in \cite{ghani_kupke_lambert_forsberg_compositional_treatment_iterated_open_games,hedges_morphisms_open_games}.
We also investigate a family of monoids and comonoids in the category of open games that are closely related to the fixpoint agent, which similarly narrowly fails to define a family of commutative bialgebras.

Both of these structures fail in precisely the same way: Morphisms which `should' equal the identity are instead equal to a certain morphism that is not related to the identity by a 2-cell.
This morphism is an open game that has the same behaviour as the identity open game \emph{in equilibrium}, but can behave differently in general.

%\textcolor{red}{Is there any way to state these results in a more positive
%light? This makes it sound like we only have negative results} 

\section{Open games}

In this section we provide a theoretically self-contained definition of open games.
However for reasons of space we find it necessary to refer the reader to \cite{hedges15b} or \cite{hedges_towards_compositional_game_theory} for motivation and further details, including the links between open games and classical game theory.

\begin{definition}
	Let $X, S, Y, R$ be sets.
	An open game $\G : (X, S) \to (Y, R)$ is defined by the following data:
	\begin{itemize}
		\item A set $\Sigma (\G)$ of \emph{strategy profiles}
		\item A \emph{play function} $\P_\G : \Sigma (\G) \times X \to Y$
		\item A \emph{coplay function} $\C_\G : \Sigma (\G) \times X \times R \to S$
		\item An \emph{equilibrium function} $\E_\G : X \times (Y
		\to R) \to \mathcal P (\Sigma (\G))$ 
	\end{itemize}
\end{definition}

A pair $(x, k) : X \times (Y \to R)$ is called a \emph{context} for $\G$, and a strategy profile $\sigma \in \E_\G (x, k)$ is called a \emph{Nash equilibrium in context} $(x, k)$.
Intuitively, a context contains the behaviour of a game's environment: $x$ says what happened in the past (that is relevant to $\G)$, and $k$ says what will happen in the future (that is relevant to $\G$) given $\G$'s behaviour.

A general open game $\G : (X, S) \to (Y, R)$ is denoted by a string diagram of the form
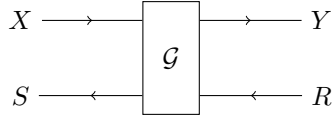
\begin{center} \begin{tikzpicture}
	\node (X) at (-2, .5) {$X$}; \node (Y) at (2, .5) {$Y$}; \node (R) at (2, -.5) {$R$}; \node (S) at (-2, -.5) {$S$};
	\node [rectangle, minimum height=1.5cm, minimum width=.75cm, draw] (G) at (0, 0) {$\G$};
	\draw [->-] (X) to (G.west |- X); \draw [->-] (G.east |- Y) to (Y); \draw [->-] (R) to (G.east |- R); \draw [->-] (G.west |- S) to (S);
\end{tikzpicture} \end{center}
In the string diagram language, a set on the left of a pair always labels a wire pointing forwards, and a set on the right of a pair always labels a wire pointing backwards.

Note that in previous work, open games $\G : (X, S) \to (Y, R)$ are defined to have a more general \emph{best response function}
\[ \mathbf B_\G : X \times (Y \to R) \to \mathcal P (\Sigma (\G) \times \Sigma (\G)) \]
rather than an equilibrium function $\E_\G$.
Any equilibrium game in this sense determines an open game in our sense by letting
\[ \E_\G (x, k) = \{ \sigma \mid (\sigma, \sigma) \in \mathbf B_\G (x, k) \} \]
While $\E_\G$ records only the Nash equilibria of a game in a context, the function $\mathbf B_\G$ records additional information about \emph{off-equilibrium best responses}, that is to say, how players might correct themselves after playing strategies that are not in equilibrium.
A Nash equilibrium is precisely a strategy profile which is a best response to itself.
The reason for using $\E_\G$ rather than $\mathbf B_\G$ in this paper will be made clear in the next section.

\begin{definition}
	An open game is called \emph{strategically trivial} if it has exactly one strategy profile and that strategy profile is a Nash equilibrium in every context.
\end{definition}

Strategically trivial open games are also called \emph{zero-player open games}.
A strategically trivial open game $(X, S) \to (Y, R)$ determines and is determined by a pair of functions $X \to Y$ and $X \times R \to S$.
Such a pair of functions is called a \emph{lens} \cite{pickering_gibbons_wu_profunctor_optics}.

\begin{definition}
	Let $f : X \to Y$ be a function.
	There are evident strategically trivial open games $(f, 1) : (X, 1) \to (Y, 1)$ and $(1, f) : (1, Y) \to (1, X)$.
\end{definition}

In the string diagram language, the open games $(f, 1)$ and $(1, f)$ are respectively denoted
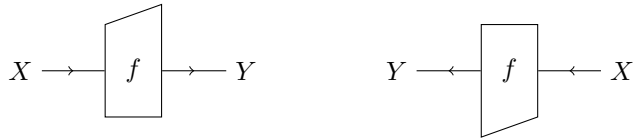
\begin{center} \begin{tikzpicture}
	\node (X1) at (0, 0) {$X$}; \node (Y1) at (3, 0) {$Y$};
	\node [trapezium, trapezium left angle=0, trapezium right angle=75, shape border rotate=90, trapezium stretches=true, minimum height=.75cm, minimum width=1.5cm, draw] (f1) at (1.5, 0) {$f$};
	\draw [->-] (X1) to (f1); \draw [->-] (f1) to (Y1);
	\node (X2) at (8, 0) {$X$}; \node (Y2) at (5, 0) {$Y$};
	\node [trapezium, trapezium left angle=0, trapezium right angle=75, shape border rotate=270, trapezium stretches=true, minimum height=.75cm, minimum width=1.5cm, draw] (f2) at (6.5, 0) {$f$};
	\draw [->-] (X2) to (f2); \draw [->-] (f2) to (Y2);
\end{tikzpicture} \end{center}

\begin{definition}
	For each set $X$, there is an evident strategically trivial open game $\epsilon_X : (X, X) \to (1, 1)$.
\end{definition}

$\epsilon_X$ is called a \emph{counit} and is denoted
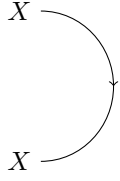
\begin{center} \begin{tikzpicture}
	\node (X1) at (0, 2) {$X$}; \node (X2) at (0, 0) {$X$};
	\draw [->-] (X1) to [out=0, in=90] (1.25, 1) to [out=-90, in=0] (X2);
\end{tikzpicture} \end{center}

Crucially, note that there is no natural strategically trivial game $(1, 1) \to (X, X)$ that could serve as a corresponding `unit'.
As such, we do not allow a backwards-pointing wire to bend around to point forwards in our diagrams.

\begin{definition}
	Let $X, Y, R$ be sets, and let $E : (Y \to R) \to \mathcal P (Y)$.
	We define an open game $\mathcal A_E : (X, 1) \to (Y, R)$, called an \emph{agent}, as follows:
	\begin{itemize}
		\item The set of strategy profiles is $\Sigma (\mathcal A_E) = X \to Y$
		\item The play function is $\P_{\mathcal A_E} (\sigma, x) = \sigma (x)$
		\item The coplay function is $\C_{\mathcal A_E} (\sigma, x, r) = *$
		\item The equilibrium function is $\E_{\mathcal A_E} (x, k) = \{ \sigma \mid \sigma (x) \in E (k) \}$
	\end{itemize}
\end{definition}

Functions $E : (Y \to R) \to \mathcal P (Y)$ are called \emph{multi-valued selection functions}, and are studied in detail in \cite{hedges_etal_selection_equilibria_higher_order_games}.
For example, let $\argmax : (Y \to \R) \to \mathcal P (Y)$ be the function defined by
\[ \argmax (k) = \{ y \mid k (y) \geq k (y') \text{ for all } y' : Y \} \]
Then $\mathcal A_{\argmax} : (X, 1) \to (Y, \R)$ is an open game representing a single decision by an agent who observes an element of $X$, and then chooses an element of $Y$ in order to maximise a real number.

Another example of a selection function is the fixpoint operator $\mathrm{fix} : (X \to X) \to \mathcal P (X)$ defined by $\mathrm{fix} (k) = \{ x \mid x = k (x) \}$.
The fixpoint selection function is the subject of a worked example in \cite{hedges_etal_selection_equilibria_higher_order_games}, where it is used to model a `Keynesian agent' whose strategic aim is to vote with the majority within a voting contest.
The open games $\mathcal A_\mathrm{fix} : (1, 1) \to (X, X)$ are the main subject of this paper.

\begin{definition}
	Let $\G : (X, S) \to (Y, R)$ and $\H : (Y, R) \to (Z, Q)$ be open games.
	We define an open game $\H \circ \G : (X, S) \to (Z, Q)$ as follows:
	\begin{itemize}
		\item The set of strategy profiles is $\Sigma (\H \circ \G) = \Sigma (\G) \times \Sigma (\H)$
		\item The play function is $\P_{\H \circ \G} ((\sigma, \tau), x) = \P_\H (\tau, \P_\G (\sigma, x))$
		\item The coplay function is $\C_{\H \circ \G} ((\sigma, \tau), x, q) = \C_\G (\sigma, x, \C_\H (\tau, \P_\G (\sigma, x), q))$
		\item The equilibrium function is
		\begin{align*}
			\E_{\H \circ \G} (x, k) = \{ (\sigma, \tau) \mid\ &\sigma \in \E_\G (x, \lambda y . \C_\H (\tau, y, k (\P_\H (\tau, y)))) \\
			&\text{ and } \tau \in \E_\H (\P_\G (\sigma, x), k) \}
		\end{align*}
	\end{itemize}
\end{definition}

\begin{definition}
	Let $\G : (X, S) \to (Y, R)$ and $\H : (X', S') \to (Y', R')$ be open games.
	We define an open game $\G \otimes \H : (X \times X', S \times S') \to (Y \times Y', R \times R')$ as follows:
	\begin{itemize}
		\item The set of strategy profiles is $\Sigma (\G \otimes \H) = \Sigma (\G) \times \Sigma (\H)$
		\item The play function is $\P_{\G \otimes \H} ((\sigma, \tau), (x, x')) = (\P_\G (\sigma, x), \P_\H (\tau, x'))$
		\item The coplay function is $\C_{\G \otimes \H} ((\sigma, \tau), (x, x'), (r, r')) = (\C_\G (\sigma, x, r), \C_\H (\tau, x', r'))$
		\item The equilibrium function is
		\begin{align*}
			\E_{\G \otimes \H} ((x, x'), k) = \{ (\sigma, \tau) \mid\ &\sigma \in \E_\G (x, \lambda y . \pi_1 (k (y, \P_\H (\tau, x')))) \\
			&\text{ and } \tau \in \E_\H (x', \lambda y' . \pi_2 (k (\P_\G (\sigma, x), y'))) \}
		\end{align*}
	\end{itemize}
\end{definition}

The operator $\otimes$ defines a monoidal product with unit $I = (1, 1)$.

These sequential and parallel composition operators correspond to sequential and simultaneous play of games.
In the string diagram language they correspond respectively to end-to-end juxtaposition with joining matching wires, and disjoint side-by-side juxtaposition.
Any string diagram that does not contain a wire bending `forwards' (in the opposite way to the counit diagram) can be consistently interpreted as an open game, given interpretations of the individual nodes \cite{hedges_coherence_lenses_open_games}.

\begin{definition}
	Let $\G, \H : (X, S) \to (Y, R)$ be open games.
	A (globular) \emph{morphism of open games} $\alpha : \G \implies \H$ consists of a function $\alpha : \Sigma (\G) \to \Sigma (\H)$ satisfying the following conditions:
	\begin{itemize}
		\item For all $\sigma : \Sigma (\G)$ and $x : X$, $\P_\G (\sigma, x) = \P_\H (\alpha (\sigma), x)$
		\item For all $\sigma : \Sigma (\G)$, $x : X$ and $r : R$, $\C_\G (\sigma, x, r) = \C_\H (\alpha (\sigma), x, r)$
		\item For all $\sigma : \Sigma (\G)$, $x : X$ and $k : Y \to R$, if $\sigma \in \E_\G (x, k)$ then $\alpha (\sigma) \in \E_\H (x, k)$
	\end{itemize}
\end{definition}

Two different, more general definitions of morphisms between pairs of open games with different types are considered in \cite{ghani_kupke_lambert_forsberg_compositional_treatment_iterated_open_games,hedges_morphisms_open_games}.
However they agree on the \emph{globular} morphisms of open games, that is those that are identity on the endpoints, which is precisely this definition.

\begin{theorem}
	There is a symmetric monoidal bicategory $\Game$ where the 0-cells are pairs of sets, the 1-cells are open games and the 2-cells are morphisms of open games.
\end{theorem}

(See \cite{stay_compact_closed_bicategories} for a clear definition of symmetric monoidal bicategories.)

%Open games form a bicategory, rather than a 2-category, because $\circ$ cannot be strictly associative:
%\begin{align*}
%	\Sigma (\G_3 \circ (\G_2 \circ \G_1)) &= (\Sigma (\G_1) \times \Sigma (\G_2)) \times \Sigma (\G_3) \\
%	&\neq \Sigma (\G_1) \times (\Sigma (\G_2) \times \Sigma (\G_3)) = \Sigma ((\G_3 \circ \G_2) \circ \G_1)
%\end{align*}
The bicategory structure allows us to talk about isomorphism and natural isomorphism of open games, where $\G \cong \H$ means that there is an isomorphism $\Sigma (\G) \cong \Sigma (\H)$ that respects play, coplay and equilibria (in both directions).
There is a `horizontal' symmetric monoidal 1-category $\Game_h$ whose morphisms are isomorphism classes of open games.
Technically, our string diagrams are valued in this 1-category and denote isomorphism classes.

\begin{proposition}[\cite{hedges_towards_compositional_game_theory}, section 2.2.13]
	Let $f : X \to Y$ be a function.
	Then there is an isomorphism of open games
	
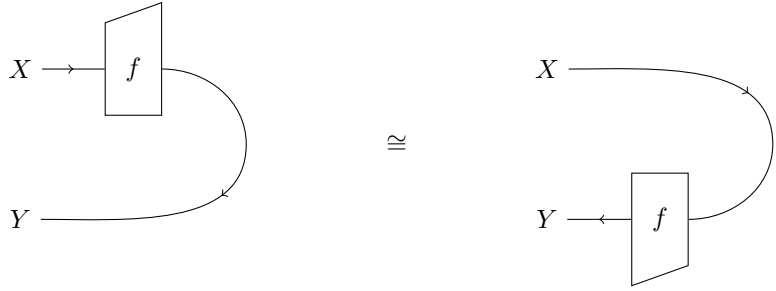
\begin{center} \begin{tikzpicture}
		\node (X1) at (0, 2) {$X$}; \node (Y1) at (0, 0) {$Y$};
		\node [trapezium, trapezium left angle=0, trapezium right angle=75, shape border rotate=90, trapezium stretches=true, minimum height=.75cm, minimum width=1.5cm, draw] (f1) at (1.5, 2) {$f$};
		\draw [->-] (X1) to (f1); \draw [->-] (f1) to [out=0, in=90] (3, 1) to [out=-90, in=0] (Y1);
		\node at (5, 1) {$\cong$};
		\node (X2) at (7, 2) {$X$}; \node (Y2) at (7, 0) {$Y$};
		\node [trapezium, trapezium left angle=0, trapezium right angle=75, shape border rotate=270, trapezium stretches=true, minimum height=.75cm, minimum width=1.5cm, draw] (f2) at (8.5, 0) {$f$};
		\draw [->-] (X2) to [out=0, in=90] (10, 1) to [out=-90, in=0] (f2); \draw [->-] (f2) to (Y2);
	\end{tikzpicture} \end{center}
\end{proposition}

\begin{definition}
	A \emph{covariant object} is a pair of the form $(X, 1)$, and a \emph{contravariant object} is a pair of the form $(1, S)$.
	We denote the former by $X^+$ and the latter $S^-$.
	Given $f : X \to Y$, we also write $f^+$ for $(f, 1)$ and $f^-$ for $(1, f)$.
\end{definition}

These respectively define covariant and contravariant monoidal functors from $(\mathbf{Set}, \times, 1)$ to $(\Game_h, \otimes, I)$.

Every object $(X, S)$ of $\Game$ is isomorphic to the
tensor product $X^+ \otimes S^-$ of a covariant object and a contravariant
object. 
%(The mapping $(X,S)\mapsto X^+\otimes S^-$
%defines a natural isomorphism.)
By lifting the unique (deleting/copying) comonoids $(X, !_X, \Delta_X)$
from the cartesian monoidal category of sets, we obtain a commutative
comonoid structure $(X^+, !_X^+, \Delta_X^+)$ on every covariant object,
and a commutative monoid structure $(X^-, !_X^-, \Delta_X^-)$ on every contravariant object.
We give these the following special syntax in the diagrammatic language:
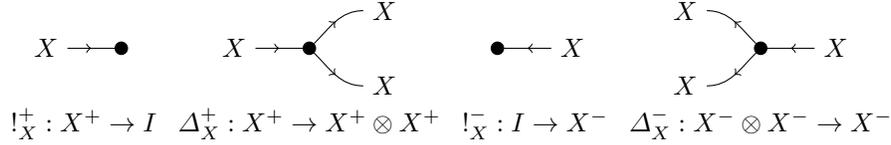
\begin{center} \begin{tikzpicture}
	\node (X1) at (0, 0) {$X$}; \node (A1) [circle, scale=.5, fill=black, draw] at (1, 0) {}; \draw [->-] (X1) to (A1);
	\node at (.5, -1) {$!_X^+ : X^+ \to I$};
	\node (X2) at (2.5, 0) {$X$}; \node (A2) [circle, scale=.5,
	fill=black, draw] at (3.5, 0) {}; \node (X3) at (4.5, .5) {$X$};
	\node (X4) at (4.5, -.5) {$X$};
	\draw [->-] (X2) to (A2); \draw [->-] (A2) to [out=45, in=180] (X3); \draw [->-] (A2) to [out=-45, in=180] (X4);
	\node at (3.5, -1) {$\Delta_X^+ : X^+ \to X^+ \otimes X^+$};
	\node (A3) [circle, scale=.5, fill=black, draw] at (6, 0) {}; \node
	(X5) at (7, 0) {$X$}; \draw [->-] (X5) to (A3);
	\node at (6.5, -1) {$!_X^- : I \to X^-$};
	\node (X6) at (8.5, .5) {$X$}; \node (X7) at (8.5, -.5) {$X$};
	\node (A4) [circle, scale=.5, fill=black, draw] at (9.5, 0) {};
	\node (X8) at (10.5, 0) {$X$};
	\draw [->-] (X8) to (A4); \draw [->-] (A4) to [out=135, in=0] (X6); \draw [->-] (A4) to [out=-135, in=0] (X7);
	\node at (9.5, -1) {$\Delta_X^- : X^- \otimes X^- \to X^-$};
\end{tikzpicture} \end{center}

\section{The fixpoint agent}

\begin{definition}
	For each set $X$ we define an open game $\eta_X = \mathcal A_\fix : I \to (X, X)$.
\end{definition}

Explicitly, $\eta_X$ is given by the following data:
\begin{itemize}
	\item The set of strategy profiles is $\Sigma (\eta_X) = X$
	\item The play function is $\P_{\eta_X} (x, *) = x$
	\item The coplay function is $\C_{\eta_X} (x, *, x') = *$
	\item The equilibrium function is $\E_{\eta_X} : 1 \times (X \to X) \to \mathcal P (X)$ is given by $\E_{\eta_X} (*, k) = \{ x : X \mid x = k (x) \}$
\end{itemize}

In the string diagram language, we denote $\eta_X$ as follows:
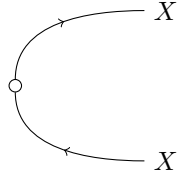
\begin{center} \begin{tikzpicture}
	\node [circle, scale=.5, draw] (A) at (0, 0) {}; \node (X1) at (2, -1) {$X$}; \node (X2) at (2, 1) {$X$};
	\draw [->-] (X1) to [out=180, in=-90] (A); \draw [->-] (A) to [out=90, in=180] (X2);
\end{tikzpicture} \end{center}

Intuitively, the fixpoint agent forces the values on its two ports to be equal \emph{in a Nash equilibrium}.
However, even if we only care about the behaviour a game in equilibrium,
the reason that an equilibrium \emph{is} an equilibrium ultimately depends on the behaviour of the game off-equilibrium.
That is, the players in the game participate in \emph{counterfactual reasoning} of the form ``What if I played a different strategy?''
This is the high level explanation of why we cannot obtain a compact closed category of open games with the fixpoint agents as its unit, even after taking a quotient to identify open games of our choosing.

With this notation, the coordination game example from \cite{hedges15b}, figure 2 is denoted
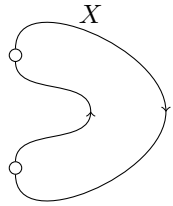
\begin{center} \begin{tikzpicture}
	\node [circle, scale=.5, draw] (A1) at (0, 0) {}; \node [circle, scale=.5, draw] (A2) at (0, 1.5) {};
	\draw [->-] (A1) to [out=90, in=-90] (1, .75) to [out=90, in=-90] (A2);
	\draw [->-] (A2) to [out=90, in=90] node [above] {$X$} (2, .75) to [out=-90, in=-90] (A1);
\end{tikzpicture} \end{center}
This represents two agents, each trying to predict the choice of the other.
The set of strategy profiles of this game is $X \times X$, and the equilibria are precisely those of the form $(x, x)$, i.e. the strategy profiles in which the agents successfully coordinate.
This game is isomorphic to the open game representation of a standard coordination game with real-valued payoffs, such as Meeting in New York.
It is closely related to the Keynesian beauty contest worked example from \cite{hedges_etal_selection_equilibria_higher_order_games}, in which three agents try to coordinate with the majority.

In a compact closed category, units and counits of
a monoidal product $X \otimes Y$ are built compositionally from the units
and counits of $X$ and $Y$. 
The counits in $\Game$ satisfy this condition, and so do the fixpoint agents:

\begin{proposition}
	$\eta_1 = \id_I : I \to I$.
\end{proposition}

\begin{proof}
	Trivial.
\end{proof}

\begin{proposition}
	Let $X$ and $Y$ be sets.
	Then $\eta_{X \times Y} : I \to (X \times Y, X \times Y)$ is
	naturally isomorphic to the following open game:
	
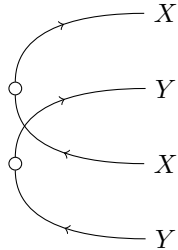
\begin{center} \begin{tikzpicture}
		\node [circle, scale=.5, draw] (A1) at (0, 2) {}; \node [circle, scale=.5, draw] (A2) at (0, 1) {};
		\node (X1) at (2, 3) {$X$}; \node (Y1) at (2, 2) {$Y$}; \node (X2) at (2, 1) {$X$}; \node (Y2) at (2, 0) {$Y$};
		\draw [->-] (X2) to [out=180, in=-90] (A1); \draw [->-] (A1) to [out=90, in=180] (X1);
		\draw [->-] (Y2) to [out=180, in=-90] (A2); \draw [->-] (A2) to [out=90, in=180] (Y1);
	\end{tikzpicture} \end{center}
\end{proposition}

\begin{proof}
	Let $\G$ be the depicted game.
	Its set $\Sigma (\G)$ of strategy profiles is naturally isomorphic to $X \times Y$.
	It is trivial to check that the play and coplay functions agree.
	Given $k : X \times Y \to X \times Y$, using the definition of $\otimes$ we have that $(x, y) \in \E_\G (*, k)$ iff $x = \pi_1 (k (x, y))$ and $y = \pi_2 (k (x, y))$.
	This is equivalent to $(x, y) = k (x, y)$, or $(x, y) \in \E_{\epsilon_{X \times Y}} (*, k)$.
\end{proof}

If we use the more general best response formulation of open games, the
previous result fails, even laxly.
This is the reason that we use the equilibrium set formulation in this paper.
(It is a rare example of a result about open games that holds in equilibrium, but can fail off-equilibrium.)

\begin{proposition}
	Let $X$ be a set. Then there is morphism of open games
	
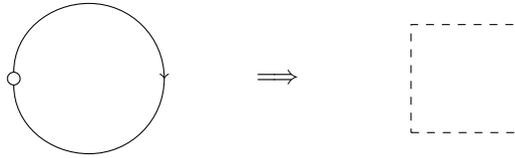
\begin{center} \begin{tikzpicture}
		\node (A1) [circle, scale=.5, draw] at (0, 0) {};
		\draw [->-] (A1) to [out=90, in=180] (1, 1) to [out=0, in=90] (2, 0) to [out=-90, in=0] (1, -1) to [out=180, in=-90] (A1);
		\node at (3.5, 0) {$\implies$};
		\node [rectangle, scale=6, dashed, draw] at (6, 0) {};
	\end{tikzpicture} \end{center}
\end{proposition}

(Note that the empty string diagram on the right hand side denotes the identity open game on $I$.)

\begin{proof}
	The sets of strategy profiles are respectively $X$ and $1$.
	Since every strategy on the left is an equilibrium, the unique function $X \to 1$ defines a morphism of open games.
\end{proof}

\begin{definition}
	For each set $X$, we define open games $\rhd_X : X^+ \to X^+$ and $\lhd_X : X^- \to X^-$, respectively denoted
	\begin{center} \begin{tikzpicture}
		\node (X1) at (0, 0) {$X$}; \node (A1) [circle, scale=.5, draw] at (1.5, 0) {}; \node (X2) at (3, 0) {$X$};
		\draw [->-] (X1) to (A1); \draw [->-] (A1) to (X2);
		\node (X3) at (5, 0) {$X$}; \node (A2) [circle, scale=.5, draw] at (6.5, 0) {}; \node (X4) at (8, 0) {$X$};
		\draw [->-] (X4) to (A2); \draw [->-] (A2) to (X3);
	\end{tikzpicture} \end{center}
	to be equal to the `snake' open games defined by
	
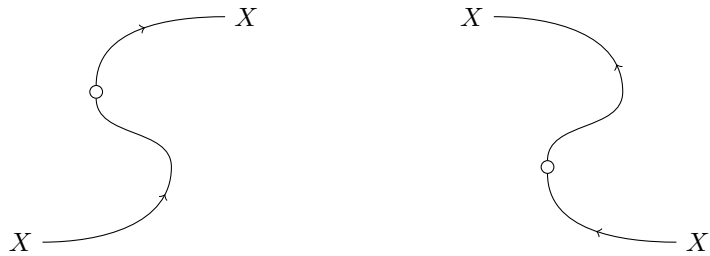
\begin{center} \begin{tikzpicture}
		\node (X1) at (0, 0) {$X$}; \node (X2) at (3, 3) {$X$};
		\node (A1) [circle, scale=.5, draw] at (1, 2) {};
		\draw [->-] (X1) to [out=0, in=-90] (2, 1) to [out=90, in=-90] (A1); \draw [->-] (A1) to [out=90, in=180] (X2);
		\node (X3) at (6, 3) {$X$}; \node (X4) at (9, 0) {$X$};
		\node (A2) [circle, scale=.5, draw] at (7, 1) {};
		\draw [->-] (X4) to [out=180, in=-90] (A2); \draw [->-] (A2) to [out=90, in=-90] (8, 2) to [out=90, in=0] (X3);
	\end{tikzpicture} \end{center}
\end{definition}

Up to natural isomorphism, $\rhd_X$ is concretely given by the data
\begin{itemize}
	\item The set of strategy profiles is $\Sigma (\rhd_X) = X$
	\item The play function is $\P_{\rhd_X} (x', x) = x'$
	\item The coplay function is $\C_{\rhd_X} (x', x, *) = *$
	\item The equilibrium function $\E_{\rhd_X} : X \times (X \to 1) \to \mathcal P (X)$ is given by $\E_{\rhd_X} (x, *) = \{ x \}$
\end{itemize}
and $\lhd_X$ is given by the data
\begin{itemize}
	\item The set of strategy profiles is $\Sigma (\lhd_X) = X$
	\item The play function is $\P_{\lhd_X} (x, *) = *$
	\item The coplay function is $\P_{\lhd_X} (x', *, x) = x'$
	\item The equilibrium function $\E_{\lhd_X} : 1 \times (1 \to X) \to \mathcal P (X)$ is given by $\E_{\lhd_X} (*, k) = \{ k (*) \}$
\end{itemize}

There are unique functions $\Sigma (\rhd_X) \to \Sigma (\id_{X^+})$ and $\Sigma (\lhd_X) \to \Sigma (\id_{X^-})$, namely $x \mapsto *$, however they fail to define morphisms of open games.
In particular, for $x' \neq x$ we have $\P_{\rhd_X} (x', x) = x' \neq x = \P_{\id_{(X, 1)}} (*, x)$.
The same argument applies to any choice of function $\Sigma (\id_{X^+}) \to \Sigma (\rhd_X)$ or $\Sigma (\id_{X^-}) \to \Sigma (\lhd_X)$.
Thus we do not obtain a compact closed bicategory \cite{stay_compact_closed_bicategories}, or even a weaker lax or colax variant of one.

The problem remains open of finding a sense in which $\rhd_X$ is related to $\id_{X^+}$ and $\lhd_X$ is related to $\id_{X^-}$.
The authors explored the following equivalence relation on the class of open games of a fixed type:

\begin{definition}
	Let $\G, \H : (X, S) \to (Y, R)$ be open games.
	Given a context $(x, k) : X \times (Y \to R)$ and strategies $\sigma : \Sigma (\G)$, $\tau : \Sigma (\H)$, we write $\sigma \sim_{(x, k)} \tau$ if $\P_\G (\sigma, x) = \H (\tau, x) =: y$ and $\C_\G (\sigma, x, k (y)) = \C_\H (\tau, x, k (y))$.
	We write $\G \sim \H$ if for every $\sigma \in \E_\G (x, k)$ there is $\tau \in \E_\H (x, k)$ with $\sigma \sim_{(x, k)} \tau$, and for every $\tau \in \E_\H (x, k)$ there is $\sigma \in \E_\G (x, k)$ with $\sigma \sim_{(x, k)} \tau$.
\end{definition}

This relation satisfies $\rhd_X \sim \id_{X^+}$ and $\lhd_X \sim \id_{X^-}$, and apparently captures the intuition that these games are `the same' in the sense that they have the same behaviour in every Nash equilibrium.
Unfortunately, $\sim$ is not compositional: there are open games $\G \sim
\G'$ and $\H \sim \H'$ for which $\H \circ \G \not\sim \H' \circ \G'$.
An interpretation of this is that morphisms of open games require behaviour
to be the same in all contexts, not just those in equilibrium.
%\textcolor{red}{I'll return to think about this. I guess a thought is that,
%even in an equilibrium the continuation contains off-equilibrium
%information}
This extra generality is crucial to making morphisms of open games form a monoidal bicategory, that is to say, to be compositional.

%Another possibility is to obtain morphisms of open games $\id_{(X, 1)} \implies \rhd_X$ and $\id_{(1, X)} \implies \lhd_X$ by generalising the definition of globular morphisms of open games.
%For example, a morphism $\G \implies \H : (X, S) \to (Y, R)$ could be a function $\Sigma (\G) \times X \times R \to \Sigma (\H)$ satisfying suitable properties.
%However, it is beyond the scope of this paper to verify whether this still results in a symmetric monoidal bicategory, or to interpret it in game-theoretic terms.
%If it is possible, it would result in a lax version of a compact closed bicategory using the fixpoint agent.

\begin{proposition}
	Let $f : X \to Y$ be a function.
	Then there is a morphism of open games
	
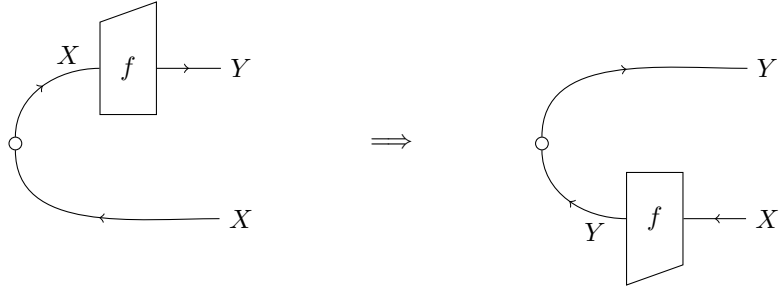
\begin{center} \begin{tikzpicture}
		\node (A1) [circle, scale=.5, draw] at (0, 0) {}; \node (X1) at (3, -1) {$X$}; \node (Y1) at (3, 1) {$Y$};
		\node [trapezium, trapezium left angle=0, trapezium right angle=75, shape border rotate=90, trapezium stretches=true, minimum height=.75cm, minimum width=1.5cm, draw] (f1) at (1.5, 1) {$f$};
		\draw [->-] (X1) to [out=180, in=-90] (A1); \draw [->-] (A1) to [out=90, in=180] node [above, near end] {$X$} (f1); \draw [->-] (f1) to (Y1);
		\node at (5, 0) {$\implies$};
		\node (A2) [circle, scale=.5, draw] at (7, 0) {}; \node (X2) at (10, -1) {$X$}; \node (Y2) at (10, 1) {$Y$};
		\node [trapezium, trapezium left angle=0, trapezium right angle=75, shape border rotate=270, trapezium stretches=true, minimum height=.75cm, minimum width=1.5cm, draw] (f2) at (8.5, -1) {$f$};
		\draw [->-] (X2) to (f2); \draw [->-] (f2) to [out=180, in=-90] node [below, near start] {$Y$} (A2); \draw [->-] (A2) to [out=90, in=180] (Y2);
	\end{tikzpicture} \end{center}
\end{proposition}

\begin{proof}
	The sets of strategy profiles of these games are respectively $X$ and $Y$.
	The function $f : X \to Y$ defines a morphism of open games.
\end{proof}

This is the first instance of a general pattern in this paper, that we can move an open game \emph{backwards} past a white node.

%\begin{proof}
%	Let $\G, \H : I \to (Y, X)$ be these pair of games.
%	Explicitly, $\G$ is given by the following data:
%	\begin{itemize}
%		\item The set of strategy profiles is $\Sigma (\G) = X$
%		\item The play function is $\P_\G (\sigma, *) = f (\sigma)$
%		\item The coplay function is $\C_\G (\sigma, *, x) = *$
%		\item The equilibrium function $\E_\G : 1 \times (Y \to X) \to \mathcal P (X)$ is given by $\E_\G (*, k) = \{ x : X \mid x = k (f (x)) \}$
%	\end{itemize}
%	Similarly, $\H$ is given by the following data:
%	\begin{itemize}
%		\item The set of strategy profiles is $\Sigma (\H) = Y$
%		\item The play function is $\P_\H (\sigma, *) = \sigma$
%		\item The coplay function is $\C_\H (\sigma, *, x) = *$
%		\item The equilibrium function $\E_\H : 1 \times (Y \to X) \to \mathcal P (Y)$ is given by $\E_\H (*, k) = \{ y : Y \mid y = f (k (y)) \}$
%	\end{itemize}
%	
%	We show that $f : \Sigma (\G) \to \Sigma (\H)$ defines a morphism of open games.
%	
%	For the play functions we have for all $\sigma : X$
%	\[ \P_\G (\sigma, *) = f (\sigma) = \P_\H (f (\sigma), *) \]
%	For the coplay functions we have for all $\sigma : X$ and $x : X$
%	\[ \C_\G (\sigma, *, x) = * = \C_\H (f (\sigma), *, x) \]
%	
%	For the equilibrium set, let $k : Y \to X$ and suppose that $x \in \E_\G (*, k)$, so $x = k (f (x))$.
%	Then $f (x) = f (k (f (x)))$, so $f (x) \in \E_\H (*, k)$.
%\end{proof}

\begin{proposition}
	For any open game $\G : X^+ \to Y^+$ between covariant objects there is a morphism of open games
	\begin{center} \begin{tikzpicture}
		\node (w1) [circle, scale=.5, draw] at (1, 0) {}; \node (X1) at (0, 0) {$X$}; \node (Y1) at (4, 0) {$Y$};
		\node [rectangle, minimum height=1.5cm, minimum width=.75cm, draw] (G1) at (2.5, 0) {$\G$};
		\draw [->-] (X1) to (w1); \draw [->-] (w1) to node [above] {$X$} (G1); \draw [->-] (G1) to (Y1);
		\node at (5.5, 0) {$\implies$};
		\node (X2) at (7, 0) {$X$}; \node [rectangle, minimum height=1.5cm, minimum width=.75cm, draw] (G2) at (8.5, 0) {$\G$};
		\node (w2) [circle, scale=.5, draw] at (10, 0) {}; \node (Y2) at (11, 0) {$Y$};
		\draw [->-] (X2) to (G2); \draw [->-] (G2) to node [above] {$Y$} (w2); \draw [->-] (w2) to (Y2);
	\end{tikzpicture} \end{center}
	For any open game $\H : Y^- \to X^-$ between contravariant objects there is a morphism of open games
	\begin{center} \begin{tikzpicture}
		\node (Y1) at (0, 0) {$Y$}; \node (H1) [rectangle, minimum height=1.5cm, minimum width=.75cm, draw] at (1.5, 0) {$\H$};
		\node (w1) [circle, scale=.5, draw] at (3, 0) {}; \node (X1) at (4, 0) {$X$};
		\draw [->-] (X1) to (w1); \draw [->-] (w1) to node [above] {$X$} (H1); \draw [->-] (H1) to (Y1);
		\node at (5.5, 0) {$\implies$};
		\node (Y2) at (7, 0) {$Y$}; \node (w2) [circle, scale=.5, draw] at (8, 0) {};
		\node (H2) [rectangle, minimum height=1.5cm, minimum width=.75cm, draw] at (9.5, 0) {$\H$}; \node (X2) at (11, 0) {$X$};
		\draw [->-] (X2) to (H2); \draw [->-] (H2) to node [above] {$Y$} (w2); \draw [->-] (w2) to (Y2);
	\end{tikzpicture} \end{center}
\end{proposition}

\begin{proof}
	For the former pair, the function $X \times \Sigma (\G) \to \Sigma (\G) \times Y$ is given by $(x, \sigma) \mapsto (\sigma, \P_\G (\sigma, x))$.
	For the latter pair, the function $\Sigma (\H) \times X \to Y \times \Sigma (\H)$ is given by $(\sigma, x) \mapsto (\C_\G (\sigma, *, x), \sigma)$.
\end{proof}

\section{Bialgebras in $\Game$}

In section 2 we defined copying operators, which are lifted from the copying comonoids in the category of sets.
In a compact closed category, the transpose of a copying operator is a dual `matching' operator, and these typically interact as a (special commutative) Frobenius algebra.
(The category of relations provides an example.)
In this section we define open games that behave like matching operators \emph{when in equilibrium}, and investigate their properties.
Surprisingly, these `imperfect' matching operators interact with the `true' copying operators not as a Frobenius algebra, but (almost) as a commutative \emph{bialgebra}.

\begin{definition}
	Let $X$ be a set.
	We define open games
	\begin{center} \begin{tikzpicture}
		\node (A1) [circle, scale=.5, draw] at (0, 0) {}; \node (X1) at (1, 0) {$X$}; \draw [->-] (A1) to (X1);
		\node at (.5, -1) {$\bangrot_X^+ : I \to X^+$};
		\node (X2) at (2.5, .5) {$X$}; \node (X3) at (2.5, -.5) {$X$}; \node (A2) [circle, scale=.5, draw] at (3.5, 0) {}; \node (X4) at (4.5, 0) {$X$};
		\draw [->-] (X2) to [out=0, in=135] (A2); \draw [->-] (X3) to [out=0, in=-135] (A2); \draw [->-] (A2) to (X4);
		\node at (3.5, -1) {$\nabla_X^+ : X^+ \otimes X^+ \to X^+$};
		\node (X5) at (6, 0) {$X$}; \node (A3) [circle, scale=.5, draw] at (7, 0) {}; \draw [->-] (A3) to (X5);
		\node at (6.5, -1) {$\bangrot_X^- : X^- \to I$};
		\node (X6) at (8.5, 0) {$X$}; \node (A4) [circle, scale=.5, draw] at (9.5, 0) {}; \node (X7) at (10.5, .5) {$X$}; \node (X8) at (10.5, -.5) {$X$};
		\draw [->-] (X7) to [out=180, in=45] (A4); \draw [->-] (X8) to [out=180, in=-45] (A4); \draw [->-] (A4) to (X6);
		\node at (9.5, -1) {$\nabla_X^- : X^- \to X^- \otimes X^-$};
	\end{tikzpicture} \end{center}
	as follows.
	In each case, the set of strategy profiles is $\Sigma = X$.
	$\bangrot_X^+$ and $\nabla_X^+$ have play functions $\P_{\bangrot_X^+} (x', *) = x'$ and $\P_{\nabla_X^+} (x', (x_1, x_2)) = x'$.
	$\bangrot_X^-$ and $\nabla_X^-$ have coplay functions $\C_{\bangrot_X^-} (x', *, *) = x'$ and $\C_{\nabla_X^-} (x', *, (x_1, x_2)) = x'$.
	The equilibrium sets are respectively
	\begin{align*}
		\E_{\bangrot_X^+} (*, *) &= X && \E_{\nabla_X^+} ((x_1, x_2), *) = \begin{cases}
			\{ x_1 \} &\text{ if } x_1 = x_2 \\
			\varnothing &\text{ otherwise}
		\end{cases} \\
		\E_{\bangrot_X^-} (*, *) &= X && \E_{\nabla_X^-} (*, k) = \begin{cases}
			\{ \pi_1 (k (*)) \} &\text{ if } \pi_1 (k (*)) = \pi_2 (k (*)) \\
			\varnothing &\text{ otherwise}
		\end{cases}
	\end{align*}
\end{definition}

\begin{proposition}
	There are morphisms of open games
	
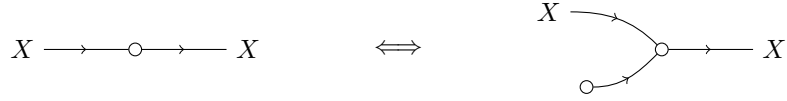
\begin{center} \begin{tikzpicture}
		\node (X1) at (0, 0) {$X$}; \node (A1) [circle, scale=.5, draw] at (1.5, 0) {}; \node (X2) at (3, 0) {$X$};
		\draw [->-] (X1) to (A1); \draw [->-] (A1) to (X2);
		\node at (5, 0) {$\iff$};
		\node (X3) at (7, .5) {$X$}; \node (A2) [circle, scale=.5, draw] at (7.5, -.5) {};
		\node (A3) [circle, scale=.5, draw] at (8.5, 0) {}; \node (X4) at (10, 0) {$X$};
		\draw [->-] (X3) to [out=0, in=135] (A3); \draw [->-] (A2) to [out=0, in=-135] (A3); \draw [->-] (A3) to (X4);
	\end{tikzpicture} \end{center}
	and
	
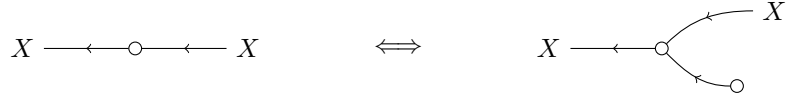
\begin{center} \begin{tikzpicture}
		\node (X1) at (0, 0) {$X$}; \node (A1) [circle, scale=.5, draw] at (1.5, 0) {}; \node (X2) at (3, 0) {$X$};
		\draw [->-] (X2) to (A1); \draw [->-] (A1) to (X1);
		\node at (5, 0) {$\iff$};
		\node (X3) at (7, 0) {$X$}; \node (A2) [circle, scale=.5, draw] at (8.5, 0) {};
		\node (X4) at (10, .5) {$X$}; \node (A3) [circle, scale=.5, draw] at (9.5, -.5) {};
		\draw [->-] (X4) to [out=180, in=45] (A2); \draw [->-] (A3) to [out=180, in=-45] (A2); \draw [->-] (A2) to (X3);
	\end{tikzpicture} \end{center}
	(which are, however, not isomorphisms).
\end{proposition}

\begin{proof}
	The top-left and top-right games have sets of strategy profiles $X$ and $X \times X$.
	The functions $x \mapsto (x, x)$ and $(x_1, x_2) \mapsto x_2$ define morphisms of open games.
	
	The bottom-left and bottom-right games have sets of strategy profiles $X$ and $X \times X$.
	The functions $x \mapsto (x, x)$ and $(x_1, x_2) \mapsto x_1$ define morphisms of open games.
\end{proof}

The distinction between $\rhd_X, \lhd_X$ and the identities means that the white structures fail to be monoids and comonoids, in precisely the same way that $\Game$ fails to be compact closed.
However, this is the only condition that fails.

%\begin{proof}
%	The top-right game $\G$ has set of strategy profiles $\Sigma (\G) = X \times X$, play function $\P_\G ((x'_1, x'_2), x) = x'_2$ and equilibrium function $\E_\G (x, *) = \{ (x, x) \}$.
%	It is easy to check that the function $\Sigma (\rhd_X) \to \Sigma (\G)$ given by $x \mapsto (x, x)$ and the function $\Sigma (\G) \to \Sigma (\rhd_X)$ given by $(x_1, x_2) \mapsto x_2$ are morphisms of open games.
%	
%	The bottom-right game $\H$ has set of strategy profiles $\Sigma (\H) = X \times X$, coplay function $\C_\H ((x'_1, x'_2), *, x) = x'_1$ and equilibrium function $\E_\H (*, k) = \{ (k (*), k (*)) \}$.
%	It is easy to check that the function $\Sigma (\lhd_X) \to \Sigma (\H)$ given by $x \mapsto (x, x)$ and the function $\Sigma (\H) \to \Sigma (\lhd_X)$ given by $(x_1, x_2) \mapsto x_1$ are morphisms of open games.
%\end{proof}

\begin{proposition}
	There are natural isomorphisms of open games
	\begin{center} \begin{tikzpicture}
		\node [circle, scale=.5, draw] (m1) at (0, .5) {}; \node [circle, scale=.5, draw] (m2) at (1, 0) {};
		\node (X1) at (-1, 1) {$X$}; \node (X2) at (-1, 0) {$X$}; \node (X3) at (-1, -1) {$X$}; \node (X4) at (2, 0) {$X$};
		\draw [->-] (X1) to [out=0, in=135] (m1); \draw [->-] (X2) to [out=0, in=-135] (m1);
		\draw [->-] (m1) to [out=0, in=135] (m2); \draw [->-] (X3) to [out=0, in=-135] (m2);
		\draw [->-] (m2) to (X4);
		\node at (4, 0) {$\cong$};
		\node [circle, scale=.5, draw] (m3) at (7, -.5) {}; \node [circle, scale=.5, draw] (m4) at (8, 0) {};
		\node (X5) at (6, 1) {$X$}; \node (X6) at (6, 0) {$X$}; \node (X7) at (6, -1) {$X$}; \node (X8) at (9, 0) {$X$};
		\draw [->-] (X5) to [out=0, in=135] (m4); \draw [->-] (X6) to [out=0, in=135] (m3);
		\draw [->-] (X7) to [out=0, in=-135] (m3); \draw [->-] (m3) to [out=0, in=-135] (m4);
		\draw [->-] (m4) to (X8);
	\end{tikzpicture} \end{center}
	and
	\begin{center} \begin{tikzpicture}
		\node [circle, scale=.5, draw] (m1) at (0, 0) {}; \node [circle, scale=.5, draw] (m2) at (1, .5) {};
		\node (X1) at (-1, 0) {$X$}; \node (X2) at (2, 1) {$X$}; \node (X3) at (2, 0) {$X$}; \node (X4) at (2, -1) {$X$};
		\draw [->-] (X2) to [out=180, in=45] (m2); \draw [->-] (X3) to [out=180, in=-45] (m2);
		\draw [->-] (m2) to [out=180, in=45] (m1); \draw [->-] (X4) to [out=180, in=-45] (m1);
		\draw [->-] (m1) to (X1);
		\node at (4, 0) {$\cong$};
		\node [circle, scale=.5, draw] (m3) at (7, 0) {}; \node [circle, scale=.5, draw] (m4) at (8, -.5) {};
		\node (X5) at (6, 0) {$X$}; \node (X6) at (9, 1) {$X$}; \node (X7) at (9, 0) {$X$}; \node (X8) at (9, -1) {$X$};
		\draw [->-] (X6) to [out=180, in=45] (m3); \draw [->-] (X7) to [out=180, in=45] (m4);
		\draw [->-] (X8) to [out=180, in=-45] (m4); \draw [->-] (m4) to [out=180, in=-45] (m3);
		\draw [->-] (m3) to (X5);
	\end{tikzpicture} \end{center}
\end{proposition}

%\begin{proof}
%	The former pair both have set of strategy profiles $X \times X$, play function $\P ((x'_1, x'_2), (x_1, x_2, x_3)) = x'_2$ and equilibrium function
%	\[ \E ((x_1, x_2, x_3), *) = \begin{cases}
%		\{ (x_1, x_1) \} &\text{ if } x_1 = x_2 = x_3 \\
%		\varnothing &\text{ otherwise}
%	\end{cases} \]
%	The latter pair both have set of strategy profiles $X \times X$, coplay function
%	\[ \C ((x'_1, x'_2), *, (x_1, x_2, x_3)) = x'_1 \]
%	and equilibrium function
%	\[ \E (*, k) = \begin{cases}
%		\{ (\pi_1 (k (*)), \pi_1 (k (*))) \} &\text{ if } \pi_1 (k (*)) = \pi_2 (k (*)) = \pi_3 (k (*)) \\
%		\varnothing &\text{ otherwise}
%	\end{cases} \]
%\end{proof}

\begin{proposition}
	There are natural isomorphisms of open games
	\begin{center} \begin{tikzpicture}
		\node (X1) at (0, .5) {$X$}; \node (X2) at (0, -.5) {$X$}; \node (X3) at (3, 0) {$X$};
		\node (A1) [circle, scale=.5, draw] at (2, 0) {};
		\draw [->-] (X1) to [out=0, in=180] (1, -.5) to [out=0, in=-135] (A1); \draw [->-] (X2) to [out=0, in=180] (1, .5) to [out=0, in=135] (A1);
		\draw [->-] (A1) to (X3);
		\node at (5, 0) {$\cong$};
		\node (X4) at (7, .5) {$X$}; \node (X5) at (7, -.5) {$X$}; \node (X6) at (9.5, 0) {$X$};
		\node (A2) [circle, scale=.5, draw] at (8.5, 0) {};
		\draw [->-] (X4) to [out=0, in=135] (A2); \draw [->-] (X5) to [out=0, in=-135] (A2); \draw [->-] (A2) to (X6);
	\end{tikzpicture} \end{center}
	and
	\begin{center} \begin{tikzpicture}
		\node (X1) at (0, 0) {$X$}; \node (X2) at (3, .5) {$X$}; \node (X3) at (3, -.5) {$X$};
		\node (A1) [circle, scale=.5, draw] at (1, 0) {};
		\draw [->-] (X2) to [out=180, in=0] (2, -.5) to [out=180, in=-45] (A1); \draw [->-] (X3) to [out=180, in=0] (2, .5) to [out=180, in=45] (A1);
		\draw [->-] (A1) to (X1);
		\node at (5, 0) {$\cong$};
		\node (X4) at (7, 0) {$X$}; \node (X5) at (9.5, .5) {$X$}; \node (X6) at (9.5, -.5) {$X$};
		\node (A2) [circle, scale=.5, draw] at (8, 0) {};
		\draw [->-] (X5) to [out=180, in=45] (A2); \draw [->-] (X6) to [out=180, in=-45] (A2); \draw [->-] (A2) to (X4);
	\end{tikzpicture} \end{center}
\end{proposition}

Next we show that despite the white structures not being monoids and comonoids, the black and white structures interact as a lax bialgebra; see for example \cite{bonchi14} for a graphical presentation of the bialgebra axioms.

\begin{proposition}
	There are morphisms of open games
	\begin{center} \begin{tikzpicture}
		\node (w1) [circle, scale=.5, draw] at (0, 0) {}; \node (b1) [circle, scale=.5, fill=black, draw] at (1, 0) {};
		\node (X1) at (2, .5) {$X$}; \node (X2) at (2, -.5) {$X$};
		\draw [->-] (w1) to (b1); \draw [->-] (b1) to [out=45, in=180] (X1); \draw [->-] (b1) to [out=-45, in=180] (X2);
		\node at (3, 0) {$\implies$};
		\node (w2) [circle, scale=.5, draw] at (4, .5) {}; \node (w3) [circle, scale=.5, draw] at (4, -.5) {};
		\node (X3) at (5, .5) {$X$}; \node (X4) at (5, -.5) {$X$};
		\draw [->-] (w2) to (X3); \draw [->-] (w3) to (X4);
		\node (b1) [circle, scale=.5, fill=black, draw] at (7, 0) {}; \node (w1) [circle, scale=.5, draw] at (8, 0) {};
		\node (X1) at (9, .5) {$X$}; \node (X2) at (9, -.5) {$X$};
		\draw [->-] (X1) to [out=180, in=45] (w1); \draw [->-] (X2) to [out=180, in=-45] (w1); \draw [->-] (w1) to (b1);
		\node at (10, 0) {$\implies$};
		\node (b2) [circle, scale=.5, fill=black, draw] at (11, .5) {}; \node (b3) [circle, scale=.5, fill=black, draw] at (11, -.5) {};
		\node (X3) at (12, .5) {$X$}; \node (X4) at (12, -.5) {$X$};
		\draw [->-] (X3) to (b2); \draw [->-] (X4) to (b3);
	\end{tikzpicture} \end{center}
	\begin{center} \begin{tikzpicture}
		\node (b1) [circle, scale=.5, fill=black, draw] at (1, 0) {}; \node (w1) [circle, scale=.5, draw] at (2, 0) {};
		\node (X1) at (0, .5) {$X$}; \node (X2) at (0, -.5) {$X$};
		\draw [->-] (w1) to (b1); \draw [->-] (b1) to [out=135, in=0] (X1); \draw [->-] (b1) to [out=-135, in=0] (X2);
		\node at (3, 0) {$\implies$};
		\node (w2) [circle, scale=.5, draw] at (5, .5) {}; \node (w3) [circle, scale=.5, draw] at (5, -.5) {};
		\node (X3) at (4, .5) {$X$}; \node (X4) at (4, -.5) {$X$};
		\draw [->-] (w2) to (X3); \draw [->-] (w3) to (X4);
		\node (w1) [circle, scale=.5, draw] at (8, 0) {}; \node (b1) [circle, scale=.5, fill=black, draw] at (9, 0) {};
		\node (X1) at (7, .5) {$X$}; \node (X2) at (7, -.5) {$X$};
		\draw [->-] (X1) to [out=0, in=135] (w1); \draw [->-] (X2) to [out=0, in=-135] (w1); \draw [->-] (w1) to (b1);
		\node at (10, 0) {$\implies$};
		\node (b2) [circle, scale=.5, fill=black, draw] at (12, .5) {}; \node (b3) [circle, scale=.5, fill=black, draw] at (12, -.5) {};
		\node (X3) at (11, .5) {$X$}; \node (X4) at (11, -.5) {$X$};
		\draw [->-] (X3) to (b2); \draw [->-] (X4) to (b3);
	\end{tikzpicture} \end{center}
\end{proposition}

\begin{proposition}
	There are morphisms of open games
	\begin{center} \begin{tikzpicture}
		\node (w1) [circle, scale=.5, draw] at (0, 0) {}; \node (b1) [circle, scale=.5, fill=black, draw] at (1.5, 0) {};
		\draw [->-] (w1) to node [above] {$X$} (b1);
		\node at (2.5, 0) {$\implies$};
		\node [rectangle, scale=6, dashed, draw] at (4, 0) {};
		\node (b1) [circle, scale=.5, fill=black, draw] at (6.5, 0) {}; \node (w1) [circle, scale=.5, draw] at (8, 0) {};
		\draw [->-] (w1) to node [above] {$X$} (b1);
		\node at (9, 0) {$\implies$};
		\node [rectangle, scale=6, dashed, draw] at (10.5, 0) {};
	\end{tikzpicture} \end{center}
\end{proposition}

\begin{proposition}
	There are morphisms of open games
	\begin{center} \begin{tikzpicture}
		\node (X1) at (0, .5) {$X$}; \node (X2) at (0, -.5) {$X$};
		\node (w1) [circle, scale=.5, draw] at (1, 0) {}; \node (b1) [circle, scale=.5, fill=black, draw] at (2, 0) {};
		\node (X3) at (3, .5) {$X$}; \node (X4) at (3, -.5) {$X$};
		\draw [->-] (X1) to [out=0, in=135] (w1); \draw [->-] (X2) to [out=0, in=-135] (w1); \draw [->-] (w1) to (b1);
		\draw [->-] (b1) to [out=45, in=180] (X3); \draw [->-] (b1) to [out=-45, in=180] (X4);
		\node at (5, 0) {$\implies$};
		\node (X5) at (7, .5) {$X$}; \node (X6) at (7, -.5) {$X$};
		\node (b2) [circle, scale=.5, fill=black, draw] at (8, .5) {}; \node (b3) [circle, scale=.5, fill=black, draw] at (8, -.5) {};
		\node (w2) [circle, scale=.5, draw] at (9.5, .5) {}; \node (w3) [circle, scale=.5, draw] at (9.5, -.5) {};
		\node (X7) at (10.5, .5) {$X$}; \node (X8) at (10.5, -.5) {$X$};
		\draw [->-] (X5) to (b2); \draw [->-] (X6) to (b3);
		\draw [->-] (b2) to [out=45, in=135] (w2); \draw [-] (b2) to [out=-45, in=135] (w3);
		\draw [-] (b3) to [out=45, in=-135] (w2); \draw [->-] (b3) to [out=-45, in=-135] (w3);
		\draw [->-] (w2) to (X7); \draw [->-] (w3) to (X8);
	\end{tikzpicture} \end{center}
	and
	\begin{center} \begin{tikzpicture}
		\node (X1) at (0, .5) {$X$}; \node (X2) at (0, -.5) {$X$};
		\node (b1) [circle, scale=.5, fill=black, draw] at (1, 0) {}; \node (w1) [circle, scale=.5, draw] at (2, 0) {};
		\node (X3) at (3, .5) {$X$}; \node (X4) at (3, -.5) {$X$};
		\draw [->-] (X3) to [out=180, in=45] (w1); \draw [->-] (X4) to [out=180, in=-45] (w1); \draw [->-] (w1) to (b1);
		\draw [->-] (b1) to [out=135, in=0] (X1); \draw [->-] (b1) to [out=-135, in=0] (X2);
		\node at (5, 0) {$\implies$};
		\node (X5) at (7, .5) {$X$}; \node (X6) at (7, -.5) {$X$};
		\node (w2) [circle, scale=.5, draw] at (8, .5) {}; \node (w3) [circle, scale=.5, draw] at (8, -.5) {};
		\node (b2) [circle, scale=.5, fill=black, draw] at (9.5, .5) {}; \node (b3) [circle, scale=.5, fill=black, draw] at (9.5, -.5) {};
		\node (X7) at (10.5, .5) {$X$}; \node (X8) at (10.5, -.5) {$X$};
		\draw [->-] (X7) to (b2); \draw [->-] (X8) to (b3);
		\draw [->-] (b2) to [out=135, in=45] (w2); \draw [-] (b2) to [out=-135, in=45] (w3);
		\draw [-] (b3) to [out=135, in=-45] (w2); \draw [->-] (b3) to [out=-135, in=-45] (w3);
		\draw [->-] (w2) to (X5); \draw [->-] (w3) to (X6);
	\end{tikzpicture} \end{center}
\end{proposition}

In each of the previous propositions, the morphism is given by the copying function $X \to X \times X$ or deleting function $X \to 1$ in the category of sets.

\bibliographystyle{plainurl}
\bibliography{prediction-postselection}

\end{document}